\documentclass{article}
\usepackage{amsthm}
\newtheorem{theorem}{Theorem}
\newtheorem{lemma}{Lemma}
\newtheorem{corollary}{Corollary}
\usepackage{amsmath}
\usepackage{algpseudocode}
\usepackage[utf8]{inputenc}
\usepackage{color}
\usepackage{graphicx}

\title{Progress on Constructing Phylogenetic Networks for Languages}
\author{Tandy Warnow, Steven N. Evans, and Luay Nakhleh }
\date{}

\begin{document}
\maketitle

\begin{abstract}
In  2006, Warnow, Evans, Ringe, and Nakhleh proposed a stochastic model (hereafter, the WERN 2006 model) of multi-state linguistic character evolution that allowed for homoplasy and borrowing. 
They proved that if there is no borrowing between languages and homoplastic states are known in advance, then the phylogenetic tree  of a set of languages is statistically identifiable under this model, and they
 presented statistically consistent methods for estimating these phylogenetic trees.
However, they  left open the question of whether a phylogenetic network -- which would explicitly model borrowing between languages that are in contact -- can be estimated under the model of character evolution.
Here, we establish that under some mild additional constraints on the WERN 2006 model, the phylogenetic network topology is statistically identifiable, and  
we present  algorithms to infer the phylogenetic network.  
We discuss the ramifications for linguistic phylogenetic network estimation in practice, and suggest directions for future research.
\end{abstract}

\maketitle

\clearpage

\clearpage

\section{Introduction}

The evolutionary history of  a collection of languages  is fundamental to many questions in historical linguistics,  
including the   reconstruction of proto-languages,  estimates of dates   for diversification of languages,  and determination of the geographical and temporal origins of Indo-Europeans \cite{gray2009language,haak2015massive,chang2015ancestry}.
These phylogenetic trees can be estimated from linguistic characters, including  morphological, typological, phonological,  and lexical characters \cite{dunn2005structural,nichols2008tutorial,calude2016typology,goldstein2020indo,goldstein2022correlated}.
There are many methods for estimating phylogenetic trees, including parsimony criteria, distance-based methods, and likelihood-based techniques based on parametric models of trait evolution, 
and 
the relative strengths of these methods and how they depend on the properties of the data have been explored using both real-world and simulated datasets 
\cite{ringe2002indo,rexova2003cladistic,McMahon-McMahon-2006,barbancon-diachronica2013}.

Yet it is well known that languages do not always evolve purely via descent, with ``borrowing" between languages requiring an extension of the Stammbaum model to a model that explicitly acknowledges  exchange between languages \cite{nakhleh2005perfect,atkinson2005words,boc2010classification,nelson2011networks,skelton2015borrowing}. 
One graphical model that has been used explicitly for  language evolution  is  
composed of an underlying genetic tree on top of which there are  additional contact edges allowing for borrowing between communities that are in contact  \cite{nakhleh2005perfect,boc2010classification}.
This type of graphical model has been studied in the computational phylogenetics literature, where it is referred to as a ``tree-based phylogenetic network"  \cite{francis2015phylogenetic}.
The estimation of phylogenetic networks is very challenging, both for statistical reasons (i.e., potential non-identifiability) and computational reasons  (see discussion in \cite{cao2019empirical}); 
although   tree-based phylogenetic networks are a restricted subclass of phylogenetic networks, there are still substantial challenges in estimating these phylogenetic networks, as discussed in \cite{gambette2012quartets,keijsper2014reconstructing}.

As difficult as it is to estimate a tree-based phylogenetic network, the estimation of  a dialect continuum represents an even larger challenge, and the interpretation of a dialect continuum is also difficult \cite{nichols2008tutorial,jacques2019save}. However, at least for language families such as Indo-European,  tree-based phylogenetic networks may suffice \cite{nakhleh2005perfect}, and hence are the focus of this paper.

The inference of phylogenetic networks  depends on the graphical model (i.e., tree, tree-based phylogenetic network, etc.) and also on the stochastic model of character evolution.
Examples of relevant character evolution models include the Stochastic Dollo with Lateral Transfer model in \cite{kelly2017lateral}, which models presence/absence of cognate classes (i.e., binary characters) with borrowing, and a model for multi-state character evolution in \cite{warnow2006stochastic}, which also allows for borrowing. 
When the phylogenetic network is tree-based, we
may seek to estimate just the genetic tree (i.e., the tree in the tree-based phylogenetic network) or we can seek to estimate the entire topology of the phylogenetic network itself, which would include the location of the contact edges.  

In this study, we address the challenge of estimating  the phylogenetic network  topology under an extension  of the model proposed in \cite{warnow2006stochastic}, which we will refer to as  the  WERN 2006 model to acknowledge the four authors of the model (Warnow, Evans, Ringe, and Nakhleh).
In the WERN 2006 model, the graphical model is a tree-based phylogenetic network so that the underlying genetic tree is rooted and binary and the non-tree edges represent contact between language groups and are bidirectional. 
Characters can evolve  down the underlying genetic tree or can use one or more contact edges. 
However, if a character evolves using a contact edge so that a state is borrowed into a lineage via that contact edge, then the borrowed state replaces the state already in the lineage. Thus, 
every character evolves down some rooted tree contained within the rooted network.
The WERN 2006 model includes numeric parameters that govern the probability of change, and these  parameters depend on the type of character, which may be phonological, morphological, or lexical.  
While the phonological characters have two states, $0$ and $1$, indicating presence-absence of a sound change and $0$ indicating the ancestral state, the other characters can exhibit any number of states on the languages, and so are called ``multi-state" characters. 
The WERN 2006 model allows for homoplasy in  character evolution (i.e., parallel evolution or back-mutation, see Figure \ref{fig:luay-figure-2}), provided that the homoplastic character states are known (in other words, we know which  character states can arise as a result of either parallel evolution or back-mutation).

\begin{figure}[ht]
\centering 
 \includegraphics[height=6cm]{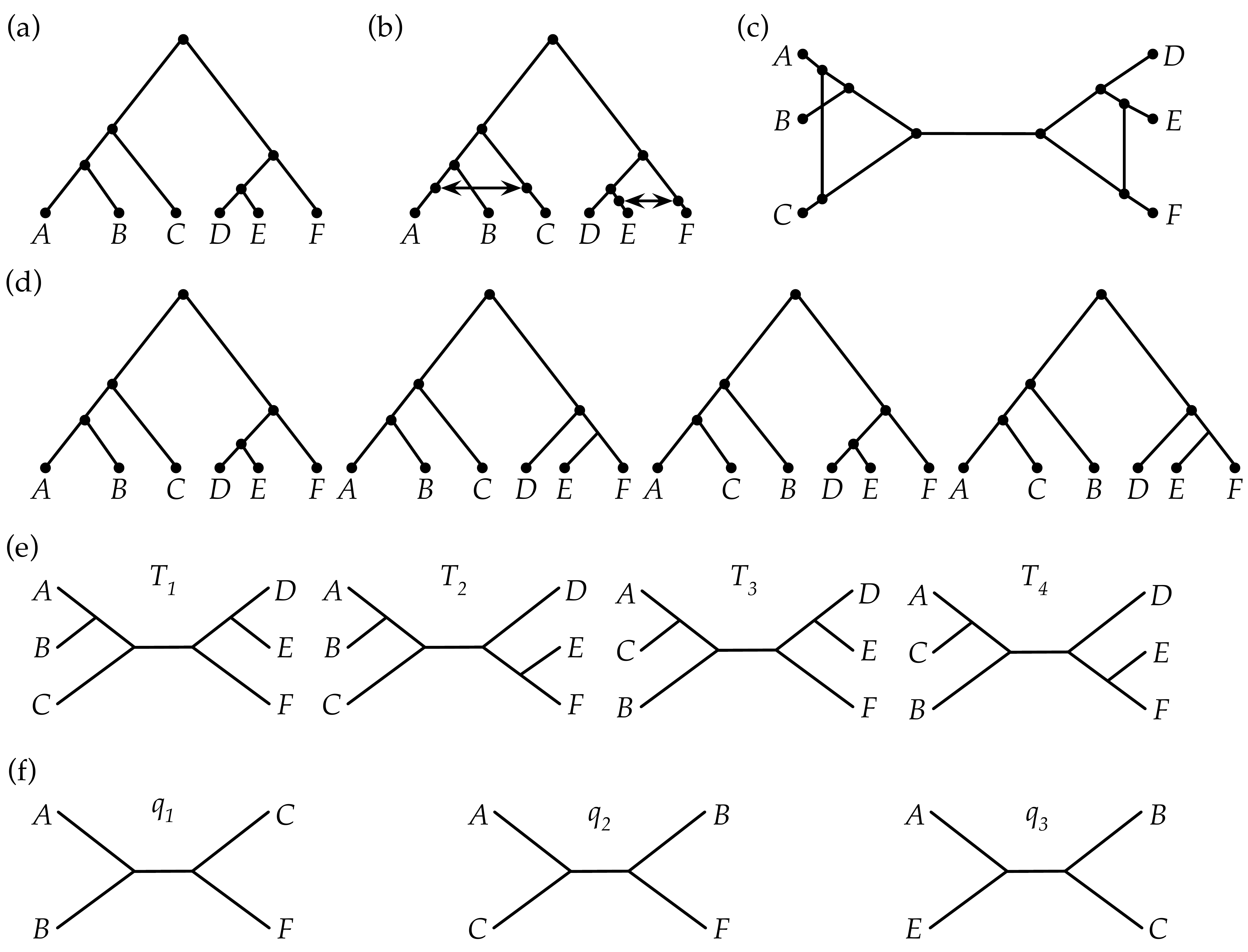}
\caption{Panel (a) shows a genetic tree with leafset $\{A,B,C,D,E,F\}$. Panel (b) shows the tree-based phylogenetic network formed by adding two contact edges   to the genetic tree. Panel (c) shows an  unrooted version of the rooted network in (b). Panel (d) shows all four  rooted trees contained inside the rooted network from (b), with the first being the genetic tree from (a). Panel (e) shows the unrooted versions of the rooted trees in (d). Panel (f)   shows three quartet trees; $q_1$ is displayed in $T_1$ and $T_2$ but not in $T_3$ or $T_4$, $q_2$ is displayed in $T_3$ and $T_4$ but not in trees $T_1$ or $T_2$, and $q_3$ is not displayed in any of these trees. Because $q_1$ and $q_2$ are each displayed by at least one tree in the network, the set $Q(N_r)$ 
will contain both $q_1$ and $q_2$, but will not contain $q_3$.}
\label{fig:luay-figure}
\end{figure}

 \begin{figure}[ht]
\centering 
 \includegraphics[height=3cm]{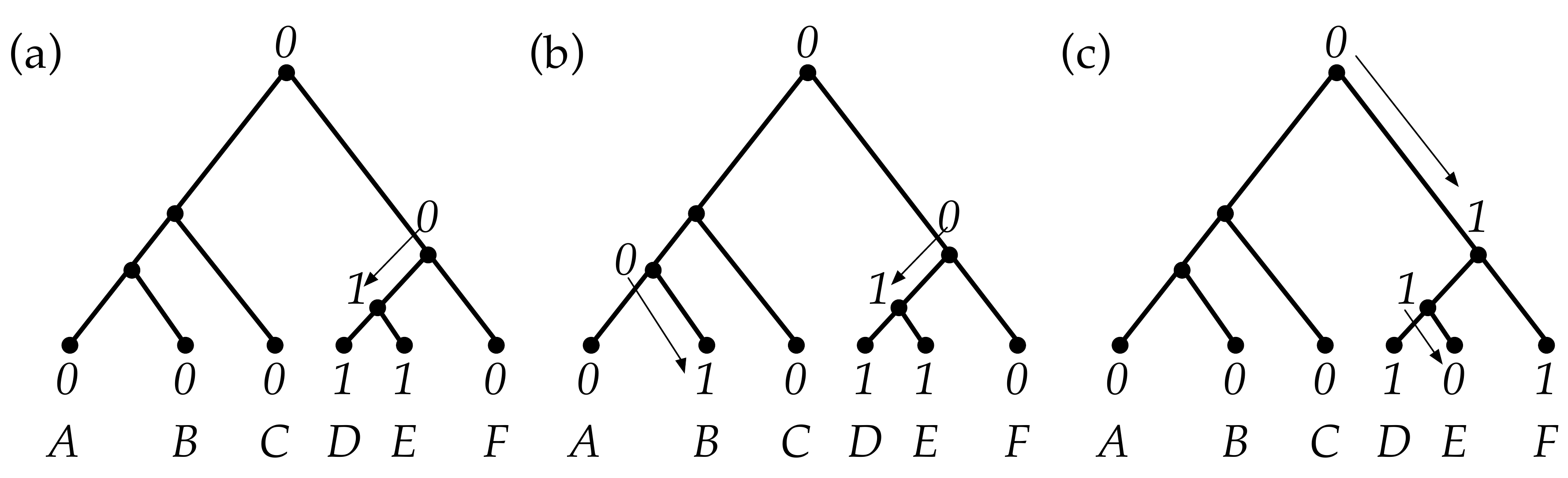}
\caption{Character evolution on a rooted tree. Panel (a) shows evolution without any homoplasy, panel (b) shows homoplasy due to parallel evolution (i.e., two $0\rightarrow 1$ transitions), and panel (c) shows homoplasy due to back-mutation (note the $1\rightarrow 0$ transition, where $0$ is the ancestral state).}
\label{fig:luay-figure-2}
\end{figure} 

Our WERN 2023 model modifies the  WERN 2006 model  as follows.
First, under the WERN 2023  model, we allow for any number of homoplastic states, as long as these states are known in advance.
We require that   the probability of homoplasy for the root state be strictly less than $1$ for all non-binary characters. 
We also allow for some characters to not exhibit any homoplasy, but the probability of a character being homoplasy-free is a parameter that can be any value $x$ with $0 \leq x \leq 1$.  
The special case where $0 < x$ means that the probability of a random character being homoplasy-free is strictly positive; when this special case holds, we will be able to use this information fruitfully.

In this article, we show  that we can estimate the unrooted topology of any WERN 2023 model phylogenetic network in a statistically consistent manner, provided that the cycles in the phylogenetic network are vertex-disjoint (which will ensure that the phylogenetic 
 network is level-1  \cite{choy2004computing,gusfield2004optimal}) and each cycle contains at least six vertices. 
The key to constructing these unrooted topologies is the inference of the unrooted quartet trees displayed by trees contained within the phylogenetic network, and these can be easily constructed from the fact that we have identifiable homoplasy.
Finally,  we also show that if homoplasy-free characters have positive probability, then we can identify the rooted topology of such a phylogenetic network.


The rest of the article is organized as follows. 
In Section \ref{sec:foundations}, we give a high-level description of the new model we propose, followed by an algorithm for   estimating the unrooted phylogenetic network and in Section \ref{sec:root-network} we present an algorithm for rooting that unrooted topology.  
We state the theoretical guarantees for the algorithms, but leave the proofs   in the appendix. 
In Section \ref{sec:practical} we discuss  the implications for the theoretical results we provide and the issues when trying  to estimate these phylogenetic networks in practice.
We conclude in Section \ref{sec:future} with a discussion of future work.

\section{Mathematical Foundations}
\label{sec:foundations}

This section introduces the basic mathematical concepts and results, but we direct the interested reader to \cite{warnow2017computational} and \cite{gusfield2014recombinatorics} for additional context. 

\subsection{Basic terminology}

The tree-based rooted phylogenetic networks $N$ we consider are formed by taking a rooted binary tree $T$ (with root $r$) and adding edges to the tree (see Figure \ref{fig:luay-figure}) so that no two cycles share any vertices. 
The edges within the rooted tree are directed away from the root towards the leaves, but the additional edges represent borrowing and so are bi-directional.
\textcolor{black}{
Each cycle in the rooted network thus has one bidirectional edge and the two nodes at the endpoint of this edge are referred to as the ``bottom nodes'' of the cycle. Note that the bottom nodes have indegree two and all other nodes have indegree one.}
To ensure identifiability, throughout this article we will constrain the phylogenetic network topology so that the smallest cycle in the unrooted network has at least six vertices.

We let $\mathcal{L}$ denote the set of languages for which we wish to construct the true phylogenetic network, $N$.
We use linguistic characters  to estimate this network,  and let $\alpha(L)$ denote the state of language $L$ for character $\alpha$.
Recall that  we say that a character exhibits {\em homoplasy} on a tree $T$ if it is not possible to assign labels to the internal vertices so that the character evolves without back mutation or parallel evolution (Figure \ref{fig:luay-figure-2}).
Furthermore, every rooted network defines a set of   rooted trees (Figure \ref{fig:luay-figure}) and every character evolves down one of the trees within the network.  
We say that a character evolves without homoplasy on a network if it is homoplasy-free on at least one of the trees inside the network; conversely,  a character exhibits homoplasy on a phylogenetic network if it exhibits homoplasy on every tree within the network.


We can also consider the rooted trees in the network  as unrooted trees, in which case they can be used to define quartet trees.
Thus, we will say that the unrooted tree $T$ displays a quartet tree 
$uv|xy$ if  $T$  has an edge $e$ that separates leaves $u,v$ from leaves $x,y$ (see Figure \ref{fig:luay-figure}). 
The set of all quartet trees displayed by any tree contained inside the network $N$ is referred to as $Q(N_r)$.
\textcolor{black}{
Note that this is \textit{not the same} as the set of all quartet trees displayed by the unrooted version of the network $N$, which is denoted by $Q(N)$ in \cite{gambette2012quartets}.}

\section{Constructing the unrooted network topology}
\label{sec:construct-unrooted-topology}

In this paper, the phylogenetic network $N$ consists of an underlying genetic tree on top of which there are borrowing edges, the cycles that are created have at least six vertices and are vertex disjoint, and
we assume that the characters  evolve down $N$   under the WERN 2023 model.
Here we describe a method that is based on computing quartet trees for constructing the unrooted topology of the phylogenetic network.

\subsection{{\em Quartet Tree-Calculator} ({\em QTC}): Constructing $Q(N_r)$}
We begin with a description of the {\em QTC} method ({\em Quartet Tree Calculator}) for computing quartet trees.
Recall that we assume we know which of the states are homoplastic.
Let $\alpha$ be a character and assume states $1$ and $2$ are both non-homoplastic.  Now suppose that we have four languages $a,b,c,d$ such that  $\alpha(a)=\alpha(b)=1$ and $ \alpha(c) = \alpha(d)=2$.
Then, we add quartet tree $ab|cd$ to our estimate  of $Q(N_r)$.
We compute these quartet trees for every character $\alpha$ in turn, thus defining a set of quartet trees that we will refer to as $Q$, the output of {\em QTC}.

\begin{theorem}
Let $N$ be a rooted phylogenetic network, and let characters evolve down $N$ under the WERN 2023 model, and let $Q$ be the output of {\em QTC}.
Then every quartet tree in $Q$ will be in $Q(N_r)$.
Furthermore, 
 as the number of characters increase, with probability converging to $1$, every quartet tree in $Q(N_r)$ will appear in   $Q$.
Thus, {\em QTC} is a statistically consistent estimator of $Q(N_r)$.
\label{thm:QTC-consistent}
\end{theorem}
The proof of this theorem is given in the appendix.


\subsection{{\em Quartet-Based Topology Estimator}}

We now present  {\em  QBTE} ({\em Quartet-Based Topology Estimator}), our method for constructing an unrooted network topology, using the quartet trees calculated using {\em QTC}.

By Theorem \ref{thm:QTC-consistent},   {\em QTC}   will return $Q(N_r)$ with probability going to $1$ as the number of characters increases. 
Hence, to estimate the unrooted topology of a phylogenetic network $N$, it suffices to use a method that can take unrooted quartet trees as input,
provided that it is guaranteed to return the unrooted topology of $N$ when given $Q(N_r)$.

\textcolor{black}{
A natural candidate is the algorithm from Section 7.1 of
\cite{gambette2012quartets}, which correctly constructs the unrooted topology of level-1 networks $N$ given $Q(N)$, and does so in $O(n^4)$ time, where $n$ is the number of leaves in the network $N$.
However, $Q(N_r)$ is in general a proper subset of $Q(N)$, and so we cannot use this algorithm as is.
Therefore, we have modified the algorithm, as we describe in the next section.}

\subsubsection{Gambette et al.'s algorithm}
\textcolor{black}{
The algorithm from \cite{gambette2012quartets}, which we refer to as ``Gambette's algorithm", makes the assumption that the input is all of $Q(N)$, where $N$ is a level-1 network.
Recall that $Q(N)$ is all quartet trees displayed by the unrooted network $N$, and so contains more than just $Q(N_r)$. Here we briefly describe the approach, see \cite{gambette2012quartets} for full details.}

\textcolor{black}{
The algorithm  has two steps.
In Step 1, an unrooted tree, referred to as the ``SN-tree'', is computed from the quartet trees. The SN-tree is the maximally resolved tree that has the property that all quartet trees in the SN-tree are in the input set of quartet trees, and no resolved quartet tree in the SN-tree conflicts with an input quartet tree.  
The SN-tree is unique for a given dense set of quartet trees (i.e., a set that has at least one tree for every four leaves), and can be constructed in polynomial time.}

\textcolor{black}{
In Step 2, the polytomies in the SN-tree are used to construct cycles. 
A polytomy with degree $k$ will produce a cycle of $k$ nodes. Moreover, each neighbor of the polytomy defines a set of leaves in the SN-tree, and any node in that set can be used to label the corresponding node in the cycle.
Gambette et al.~pick one such leaf for every node in the cycle, and thus will construct a cycle (to replace the polytomy) with labels drawn from the leaves.  They then use the quartet trees to determine the order of the labels of the nodes in the cycle (by determining which pairs of labels are adjacent). This allows them to replace each polytomy by a cycle in the correct way, and hook up the cycles to the rest of the graph.}

\textcolor{black}{
Here Gambette's algorithm explicitly assumes that all quartets in $Q(N)$ are in the input, and without this assumption the algorithm can definitely fail to be correct.}

\subsubsection{Our modification to Gambette's algorithm}
\label{sec:modified-gambette}
\textcolor{black}{
Although we can use Step 1 without modification (because it only requires that the input set of quartet trees be dense), we cannot use Step 2,  since  $Q(N_r)$ is not all of $Q(N)$ in general.  
However, we can easily modify this algorithm so that it works correctly with just the quartets from $Q(N_r)$, provided that the minimum cycle length of any cycle is six.}

\textcolor{black}{
The proof goes as follows: if  some set of four leaves $a,b,c,d$ (that are labelling different nodes for a single cycle) has only one tree, then (by Lemma \ref{lemma:main}, below) none of the bottom nodes in the cycle can be labelled by any of these leaves.  Therefore, if every cycle has at least six nodes, then we can determine which of the nodes are the bottom nodes, since every other node is in a set of four leaves that has only one tree.
As a side remark, we note that there are examples of 5-cycles where we cannot figure out what the two bottom nodes are!}

\textcolor{black}{
Furthermore, once we know the bottom nodes in the cycle, we can remove those two bottom nodes, and since the cycle has length at least six, there are at least four remaining nodes.
The set of quartet trees on the remaining nodes are compatible and define a tree, and in fact define a path.  From this path, the correct way of adding in the bottom nodes (one at each end of the path) is trivial. 
This allows us to completely determine the expansion of the polytomy into a cycle, with a unique bottom pair and bidirectional edge.}

\begin{lemma}
\textcolor{black}{
Let $a,b,c,d$ be a set of four leaves that are the labels for distinct nodes in a common cycle. 
Then $Q(N_r)$ has at exactly two quartet trees on $a,b,c,d$ if and only if 
 at least one of these leaves labels a bottom node in the cycle.
 }
 \label{lemma:main}
\end{lemma}
\begin{proof}
\textcolor{black}{
The proof is by contradiction.
If exactly one of the leaves labels a bottom node
in the cycle, then without loss of generality assume that  $d$ labels a bottom node and has parent node $z$ in the tree where $z$ is not the other bottom node.  
If we do not include the bidirectional edge, then without loss of generality, the quartet tree is $ab|cd$, and the path in the quartet tree to $c$ from $d$  goes through $z$.
Now, consider the result of using the bidirectional edge by allowing $d$ to inherit from the other bottom node, and so deleting the edge $(z,d)$. Now the quartet tree on $a,b,c,d$ is
$ad|bc$.} 

\textcolor{black}{
If two of the leaves label bottom nodes, then without loss of generality the pair of bottom nodes is $c,d$.
If we do not use the bidirectional edge, then the quartet tree splits $c$ and $d$ and so we obtain $ac|bd$.
If we use the bidirectional edge in either direction, then we obtain $ab|cd$.
Hence, if at least one of the two leaves maps to a bottom node, then we obtain two quartet trees for the set $a,b,c,d$.
}

\textcolor{black}{
Now for the other direction. If the four leaves are $a,b,c,d$ and none of them label bottom nodes, then the rooted subtree on $a,b,c,d$ does not go through the bottom nodes, and so is not impacted by including the bidirectional edge. Hence, there is only one rooted quartet tree on $a,b,c,d$, and so only one unrooted quartet tree on $a,b,c,d$.
}
\end{proof}


\subsubsection{ {\em QBTE}:  constructing the unrooted network topology}
\begin{itemize}
\item Construct a set of quartet trees $Q$ from the input $M$ character dataset, using the {\em QTC} method.
\item \textcolor{black}{Use our modification to the algorithm from \cite{gambette2012quartets} applied to $Q$ to produce the level-1 phylogenetic network for the quartet trees (see Section \ref{sec:modified-gambette}).}
\end{itemize}

\begin{theorem}
The {\em QBTE (Quartet-based topology estimation)} method is statistically consistent for estimating the unrooted topology of the network $N$ under the WERN 2023 model when the rooted network  $N$ is a level-1 network where all cycles have length at least six; furthermore, {\em QBTE} runs in polynomial
time.
\label{thm:qbte-consistent}
\end{theorem}
The proof is provided in the appendix.


\section{{\em Root-Network}: Rooting an unrooted level-1 network} 

\label{sec:root-network}

Here we present {\em Root-Network}, a method for rooting an unrooted level-1 phylogenetic network.   
Thus, the input to {\em Root-Network} will be the unrooted network $N$ and the set $\mathcal{C}_0$ of homoplasy-free phonological characters that exhibit both states $0$ and $1$ at the leaves of $N$. 
If $\mathcal{C}_0$ is empty, we mark every edge as being able to include the root, and otherwise we will process the edges to determine which edges are feasible as root locations.
At the end of processing all the homoplasy-free phonological characters, 
any edge that remains is considered a feasible root location.


When an edge $e=(a,b)$ is used as the root location, it is subdivided through the introduction of a new vertex $v_e$ so that the edge $(a,b$) is replaced by a path of length two containing two edges: $(a,v_e)$ and $(v_e,b)$. 
The vertex $v_e$ is then the root of the tree that is produced. 
Since these characters in $\mathcal{C}_0$ exhibit both states  and because $0$ is the ancestral state, making $e$ contain the root is equivalent to saying that the state of $v_e$ is $0$ for every character in $\mathcal{C}_0$.
Hence, determining if $v_e$ can be the root for a given character $\alpha \in \mathcal{C}_0$   is equivalent to saying that $v_e$ can be labelled $0$ without losing the homoplasy-free property for $\alpha$.

 {\em Root-Network} determines which edges cannot contain the root by processing each character from $\mathcal{C}_0$ in turn.
All edges are initially colored green, and any edge that is discovered to not be able to contain the root for some character is colored red. 
Under the assumptions of the algorithm, at the end of the algorithm there will be at least one edge that is not colored red.
The  set of edges that are green constitutes the set of edges that can contain the root, and will be returned by the algorithm.




\paragraph{Handling cut edges. }
An edge whose deletion splits the network into two components is
referred to as a ``cut edge."
If $e$ is a cut-edge in the network, then it is easy to tell if it should be red or green. 
Removing a cut edge $e$ splits the leafset into two sets, $A$ and $B$.
If any character exhibits state $1$ on leaves in both $A$ and $B$, then $e$ must be colored red, and otherwise it remains green.
We note that it is not possible for  both $0$ and $1$ to appear on both sides of $e$,  since that is inconsistent with homoplasy-free evolution.

\paragraph{Processing edges in  cycles}
All edges that are not cut edges are in cycles, and because we are working with a level-1 network, any such edge is in exactly one cycle.
Here we show how to color the edges that are in cycles.

Let $\gamma$ be a cycle in $N$, and assume it has $k$ vertices.  If we were to remove all the edges in the cycle, the network would split into exactly $k$ components, since all cycles in $N$ are vertex-disjoint.

Consider a single character in $\mathcal{C}_0$ and the states of this character at the leaves in each of the components defined for $\gamma$.
We split  the components into three sets: the set $A(0)$ of components all of whose leaves have state $0$, the set $A(1)$ of components all of whose leaves have state $1$, and the set $A(0,1)$  the set of components where at least one leaf has state $0$ and at least one leaf has state $1$.
Each vertex in $\gamma$ belongs to exactly one component, and so we can 
label the vertices of $\gamma$. according to the  type of component they belong to
(i.e., $A(0), A(1)$, or $A(0,1)$).
\textcolor{black}{We note that $\gamma$ has at most one vertex labelled $A(0,1)$, as otherwise the character cannot evolve without homoplasy.
We use this to determine if we should recolor the edges in $\gamma$   as follows:
\begin{itemize}
\item If there is one vertex in $\gamma$ labelled $A(0,1)$, then we color red any edge incident with a vertex labelled $A(1)$.
\item If there are no vertices in $\gamma$ labelled $A(0,1)$, then we color red any edge both of whose endpoints are labelled $A(1)$.
\end{itemize}
}




We perform this processing  for  every character, thus recoloring some edges in $\gamma$ red. 
Any edge that remains green throughout this process is returned by {\em Root-Network}.



\begin{theorem}
Let $N$ be the true unrooted level-1 network and
let $\mathcal{C}_0$ denote the set of homoplasy-free phonological characters that exhibit both $0$ and $1$ at the leaves of $N$.
Rooting $N$ on any edge returned by {\em Root-Network} will produce a rooted network on which all  characters in $\mathcal{C}_0$ can evolve without homoplasy, and the  edge containing the true location of the root will be in the output returned by {\em Root-Network}.
Furthermore, when given the unrooted topology of the true phylogenetic network as input, {\em Root-Network} is a statistically consistent estimator of the root location under the assumption that the probability of homoplasy-free phonological characters is positive.
\label{thm:RootNetwork-consistent}
\end{theorem}
The proof for this theorem is in the appendix. 
As a corollary, we have:

\begin{corollary}
The two-stage method of {\em QBTE} followed by {\em Root-Network}  is statistically consistent for estimating the rooted topology of the network $N$ under the WERN 2023 model, when the rooted network  $N$ is a level-1 network where all cycles have length at least six 
and the probability of homoplasy-free phonological characters is positive. Furthermore, this two-stage method runs in polynomial
time.
\end{corollary}
The proof follows easily from Theorems \ref{thm:qbte-consistent} and
\ref{thm:RootNetwork-consistent}.

\section{Practical considerations}
\label{sec:practical}

We have described (1) {\em QBTE}, a method for constructing the unrooted topology of a level-1 phylogenetic network from characters, and (2) {\em Root-Network}, a method for rooting the resultant topology of the level-1 network. 
Each of these methods has strong theoretical guarantees of statistical consistency.
However, these guarantees do not imply good or even reasonable accuracy on finite data, such as can occur when the input is of insufficient quantity or does not evolve under the assumptions of the theorems (e.g., down a level-1 network with known homoplastic states).


Therefore, we ask: {\em what are the consequences for estimating  the network from real-world languages, given these caveats?}
It is important to realize that the guarantees for the {\em QBTE} algorithm depend on  {\em  QTC} correctly returning the entire set of quartet trees $Q(N_r)$.
Moreover, {\em QBTE} also requires that the characters evolve down a level-1 network and that every cycle has at least six nodes. 
Even if the assumptions of the character evolution are valid, so that the characters evolve down a level-1 phylogenetic network under the WERN 2023 model,  
some of the quartet trees in $Q(N_r)$ may fail to appear in the output from {\em QTC}, which will violate the requirements for {\em QBTE} to return a network.
Furthermore, if the assumptions regarding character evolution are invalid, then some of the quartet trees produced by {\em QTC} may be incorrect (e.g., they may be quartet trees not displayed in the phylogenetic network). 
Finally, it may be that the characters evolve down a  phylogenetic network that is more complex than a level-1 network.
In each of these cases, the most likely outcome is
that {\em QBTE} will fail to return anything.

Given the likely limitations of all three methods,
we consider an alternative approach. 
Instead of estimating the unrooted network topology directly, we propose to estimate the unrooted genetic tree first using quartet trees, then (if desired) root the genetic tree and add in the contact edges.   
For example, such an approach was used in \cite{nakhleh2005perfect} to produce a perfect phylogenetic network for Indo-European.

\vspace{.1in}




\noindent
{\bf Genetic Tree Estimation (heuristic)}:
\begin{itemize}
    \item  Step 1: Construct a set $Q$ of  quartet trees using the {\em QTC} technique.
    \item Step 2: Build a tree $T$ for $\mathcal{L}$   from  $Q$, using quartet amalgamation methods that construct trees on the full leafset from sets of estimated quartet trees; examples include ASTRAL \cite{mirarab2014astral}, Quartets MaxCut \cite{QMC}, and Quartet FM \cite{reaz2014accurate}, which do not require that all the quartet trees be correct, nor that the set contain a quartet tree for every four-leaf subset of the leafset.
\end{itemize}
Note that quartet amalgamation methods typically try to solve the {\em  Maximum Quartet Support Supertree} problem, where the output is a tree that agrees with as many quartet trees in the input as possible.
Because these quartet amalgamation methods will return output trees even under adverse conditions (e.g., where many quartet trees have errors), this type of approach  is guaranteed to return {\bf a} tree $T$ provided that the set $Q$ of quartet trees produced by {\em QTC} contains quartets that cover the leafset.
This condition is   much easier to achieve than what is required for our level-1 network  estimation method, {\em QBTE}.
Moreover, when the quartet amalgamation method uses polynomial time (which is true of many such methods), this approach uses polynomial time. Hence there are several empirical advantages to this  approach over {\em QBTE}.

\section{Future Work}
\label{sec:future}
This study suggests several directions for future work. 
For example,  we recognized practical limitations of {\em QBTE}, our proposed method for estimating the unrooted phylogenetic network topology: although it is  provably statistically consistent under the WERN 2023 model, assuming that the phylogenetic network is level-1, in practice it may fail to return
any network topology for a given input. 
Hence, it has limited practical use for analyzing real world data.
Therefore, the most important future work is to determine whether there are methods that are provably statistically consistent for estimating the topologies of these tree-based phylogenetic networks that are also of practical benefit. 
The approach we suggested of estimating the genetic tree first is worthwhile, but we do not yet have any proofs of statistical consistency for that estimation using quartet amalgamation methods.


Another technique that might lead to phylogenetic network estimation methods that are of practical benefit  would seek to modify the algorithms used for {\em QBTE} so that they were guaranteed to return network topologies even when the conditions for exact accuracy did not apply.
Such extensions could potentially be implemented by seeking level-1 network topologies that agreed with the maximum number of input quartet trees.

Finally, another direction for future work is to determine whether more complex graphical models (e.g., level-2 phylogenetic networks) are identifiable under the WERN 2023 model, and whether level-1   phylogenetic networks are identifiable under character evolution models that are more complex than the WERN 2023 model. 
Future work is needed to explore these different possibilities.


\section*{Acknowledgments}
\textcolor{black}{
The authors thank C\'ecile An\'e for pointing out that $Q(N_r) \neq Q(N)$ and hence our earlier version (which assumed we could use Gambette's algorithm in \cite{gambette2012quartets}) was flawed.}

\section*{Appendix}

We restate and then sketch proofs for  Theorems 1--3.\\

\noindent
{\bf Theorem 1.}
{\em 
Let $N$ be a rooted phylogenetic network, and let characters evolve down $N$ under the WERN 2023 model, and let $Q$ be the output of {\em QTC}.
Then very quartet tree in $Q$ will be in $Q(N_r)$.
Furthermore, 
 as the number of characters increase, with probability converging to $1$, every quartet tree in $Q(N_r)$ will appear in   $Q$.
Thus, {\em QTC} is a statistically consistent estimator of $Q(N_r)$.}
\begin{proof}
We begin by showing that every quartet tree placed in $Q$ is also in $Q(N_r)$. 
Recall that quartet tree $uv|xy$ is included in $Q$ if and only if some character $\alpha$ is found such that
$\alpha(u)=\alpha(v) \neq \alpha(x)=\alpha(y)$ and the states $\alpha(u),\alpha(x)$ are non-homoplastic. This character evolves down some tree $T$ contained inside the network. 
Moreover, 
since the states exhibited at $u,v,x,y$ are non-homoplastic, 
there is a path in $T$ connecting $u$ and $v$ and another path connecting $x$ and $y$ and these two paths do not share any vertices.
Hence, the quartet tree $uv|xy$ is  in $Q(N_r)$.

We now show that in the limit, every quartet tree in $Q(N_r)$ is also in $Q$.
Let $ab|cd$ be a quartet tree in $Q(N_r)$. Hence, there is a  rooted
tree $T$ contained in $N$ that induces this quartet tree (when $T$ is considered as an unrooted tree).
With positive probability, a character will evolve down $T$.
Without loss of generality, assume $a$ and $b$ are siblings in the rooted version of $T$, so that their least common ancestor, $lca_T(a,b)$, lies strictly below the root of the tree $T$.

Since $a$ and $b$ are siblings, 
there is an
edge  $e$  above $lca_{T}(a,b)$ within $T$. 
It follows that the probability that a random character  evolves down $T$, selecting a non-homoplastic state at the root, and then changing on $e$ but on no other edge in $T$, is strictly positive.
Note that for any such characters $\alpha$, we have $\alpha(a)=\alpha(b)$ and $ \alpha(c)=\alpha(d)$ where $\alpha(a)$ and $\alpha(b)$ are different and both are non-homoplastic states. 
In such a case, 
$Q$ will include quartet tree $ab|cd$.
Thus, in the limit as the number of characters increases, with probability converging to $1$, $Q$ will contain every quartet tree in $Q(N_r)$.

Since in the limit $Q \subseteq Q(N_r)$ and $Q(N_r) \subseteq Q$, it follows that $Q = Q(N_r)$ with probability converging to $1$.
\end{proof}

\noindent {\bf Theorem 2. }
{\em 
The {\em QBTE (Quartet-based topology estimation)} method is statistically consistent for estimating the unrooted topology of the network $N$ under the WERN 2023 model when the rooted network  $N$ is a level-1 network where all cycles have length at least six; furthermore, {\em QBTE} runs in polynomial
time.}

\begin{proof}
By Theorem \ref{thm:QTC-consistent}, we have shown that as the number of characters increases,  we can construct $Q(N_r)$.
\textcolor{black}{
The algorithm we provide, QBTE, has the same two-step approach as the algorithm in \cite{gambette2012quartets}, 
differing only in the second step which expands polytomies in the SN-tree into cycles.  
In \cite{gambette2012quartets}, they showed that when their input was all of $Q(N)$ they would correctly reconstruct the level-1 network $N$.
However, in our study, we only have $Q(N_r)$ available, and $Q(N_r)$ is a proper subset of $Q(N)$.
We therefore modified Gambette's algorithm, retaining the first step but modifying the second step. 
We proved in Section \ref{sec:modified-gambette} that our modification to Gambette's algorithm correctly replaces polytomies by cycles under the conditions given (i.e., where $N$ is a level-1 network, all cycles have length at least six, and $Q(N_r)$ is given as input). }
Since a tree-based network in which no two cycles share any nodes is a level-1 network,
it follows that {\em QBTE} is statistically consistent.
\textcolor{black}{
Moreover, it is trivial that {\em QBTE} runs in polynomial time, since Gambette's algorithm \cite{gambette2012quartets} is polynomial time, and the modification to Gambette's algorithm only changes the second step and the replacement is polynomial time.}
\end{proof}

\noindent{\bf Theorem 3.}
{\em 
Let $N$ be the true unrooted level-1 network and
let $\mathcal{C}_0$ denote the set of homoplasy-free phonological characters that exhibit both $0$ and $1$ at the leaves of $N$.
Rooting $N$ on any edge returned by {\em Root-Network} will produce a rooted network on which all  characters in $\mathcal{C}_0$ can evolve without homoplasy, and the  edge containing the true location of the root will be in the output returned by {\em Root-Network}.
Furthermore, when given the unrooted topology of the true phylogenetic network as input, {\em Root-Network} is a statistically consistent estimator of the root location under the assumption that the probability of homoplasy-free phonological characters is positive.}

\begin{proof}
We sketch the proof. 
It is straightforward to verify that an edge is colored red for a character $\alpha$ if and only if subdividing the edge and labelling the introduced  node by $0$ for $\alpha$ makes $\alpha$ homoplastic on every tree contained within the network. 
Furthermore, it is not hard to see that if we root the network on any edge that remains green throughout Root-Network, then all characters  in $\mathcal{C}_0$ will be homoplasy-free.
As a result,   the first part of the theorem is established.

For the second part of the theorem, if the probability of homoplasy-free phonological characters is positive, then with probability converging to $1$, for every edge in the true network, there is a character $\alpha$ that changes on the edge but on no other edge; hence, $\alpha$ will be non-constant and homoplasy-free. 
Let $e_1$ and $e_2$ be the two edges incident to the root, and suppose the 
 input set of characters contains $\alpha_1$ and $\alpha_2$  homoplasy-free characters that change on $e_1$ and $e_2$, respectively, then these two characters will mark as red every edge below $e_1$ and $e_2$. 
In the unrooted topology for $N$, the root is suppressed and edges $e_1$ and $e_2$ are merged into the same single edge, $e$. 
Hence, when {\em Root-Network} is applied to the unrooted topology for $N$, if characters $\alpha_1$ and $\alpha_2$ are in the input, then the only edge that is not colored red will be the edge $e$ containing the suppressed root.
In conclusion, since the probability of homoplasy-free phonological characters is strictly positive, as the number of such characters increase, the probability that every edge other than the root edge will be red will converge to $1$. 
Thus, {\em Root-Network} will uniquely leave the single edge containing the suppressed root green, establishing that it is statistically consistent for locating the root in the network.
\end{proof}

\bibliography{festschrift}

\begin{thebibliography}{10}

\bibitem{atkinson2005words}
Quentin Atkinson, Geoff Nicholls, David Welch, and Russell Gray.
\newblock From words to dates: water into wine, mathemagic or phylogenetic
  inference?
\newblock {\em Transactions of the Philological Society}, 103(2):193--219,
  2005.

\bibitem{barbancon-diachronica2013}
Fran{\c{c}}ois Barban{\c{c}}on, Steven~N Evans, Luay Nakhleh, Don Ringe, and
  Tandy Warnow.
\newblock An experimental study comparing linguistic phylogenetic
  reconstruction methods.
\newblock {\em Diachronica}, 30(2):143--170, 2013.

\bibitem{boc2010classification}
Alix Boc, Anna Maria~Di Sciullo, and Vladimir Makarenkov.
\newblock {Classification of the Indo-European languages using a phylogenetic
  network approach}.
\newblock In {\em Classification as a Tool for Research}, pages 647--655.
  Springer, 2010.

\bibitem{calude2016typology}
Andreea~S Calude and Annemarie Verkerk.
\newblock {The typology and diachrony of higher numerals in Indo-European: a
  phylogenetic comparative study}.
\newblock {\em Journal of Language Evolution}, 1(2):91--108, 2016.

\bibitem{cao2019empirical}
Zhen Cao, Jiafan Zhu, and Luay Nakhleh.
\newblock Empirical performance of tree-based inference of phylogenetic
  networks.
\newblock In Katharina~T. Huber and Dan Gusfield, editors, {\em 19th
  International Workshop on Algorithms in Bioinformatics (WABI 2019)}, volume
  143 of {\em Leibniz International Proceedings in Informatics (LIPIcs)}, pages
  21:1--21:13, Dagstuhl, Germany, 2019. Schloss Dagstuhl--Leibniz-Zentrum fuer
  Informatik.

\bibitem{chang2015ancestry}
Will Chang, David Hall, Chundra Cathcart, and Andrew Garrett.
\newblock Ancestry-constrained phylogenetic analysis supports the
  {Indo-European} steppe hypothesis.
\newblock {\em Language}, pages 194--244, 2015.

\bibitem{choy2004computing}
Charles Choy, Jesper Jansson, Kunihiko Sadakane, and Wing-Kin Sung.
\newblock Computing the maximum agreement of phylogenetic networks.
\newblock {\em Electronic Notes in Theoretical Computer Science}, 91:134--147,
  2004.

\bibitem{dunn2005structural}
Michael Dunn, Angela Terrill, Ger Reesink, Robert~A Foley, and Stephen~C
  Levinson.
\newblock Structural phylogenetics and the reconstruction of ancient language
  history.
\newblock {\em Science}, 309(5743):2072--2075, 2005.

\bibitem{francis2015phylogenetic}
Andrew~R Francis and Mike Steel.
\newblock Which phylogenetic networks are merely trees with additional arcs?
\newblock {\em Systematic Biology}, 64(5):768--777, 2015.

\bibitem{gambette2012quartets}
Philippe Gambette, Vincent Berry, and Christophe Paul.
\newblock Quartets and unrooted phylogenetic networks.
\newblock {\em Journal of Bioinformatics and Computational Biology},
  10(04):1250004, 2012.

\bibitem{goldstein2020indo}
David Goldstein.
\newblock {Indo-European phylogenetics with R: A tutorial introduction}.
\newblock {\em Indo-European Linguistics}, 8(1):110--180, 2020.

\bibitem{goldstein2022correlated}
David Goldstein.
\newblock Correlated grammaticalization: The rise of articles in
  {Indo-European}.
\newblock {\em Diachronica}, 39(5):658--706, 2022.

\bibitem{gray2009language}
Russell~D Gray, Alexei~J Drummond, and Simon~J Greenhill.
\newblock Language phylogenies reveal expansion pulses and pauses in {Pacific}
  settlement.
\newblock {\em Science}, 323(5913):479--483, 2009.

\bibitem{gusfield2014recombinatorics}
Dan Gusfield.
\newblock {\em {ReCombinatorics:} the algorithmics of ancestral recombination
  graphs and explicit phylogenetic networks}.
\newblock MIT Press, 2014.

\bibitem{gusfield2004optimal}
Dan Gusfield, Satish Eddhu, and Charles Langley.
\newblock Optimal, efficient reconstruction of phylogenetic networks with
  constrained recombination.
\newblock {\em Journal of Bioinformatics and Computational Biology},
  2(01):173--213, 2004.

\bibitem{haak2015massive}
Wolfgang Haak, Iosif Lazaridis, Nick Patterson, Nadin Rohland, Swapan Mallick,
  Bastien Llamas, Guido Brandt, Susanne Nordenfelt, Eadaoin Harney, Kristin
  Stewardson, et~al.
\newblock {Massive migration from the steppe was a source for Indo-European
  languages in Europe}.
\newblock {\em Nature}, 522(7555):207--211, 2015.

\bibitem{jacques2019save}
Guillaume Jacques and Johann-Mattis List.
\newblock Save the trees: Why we need tree models in linguistic reconstruction
  (and when we should apply them).
\newblock {\em Journal of Historical Linguistics}, 9(1):128--167, 2019.

\bibitem{keijsper2014reconstructing}
JCM Keijsper and RA~Pendavingh.
\newblock Reconstructing a phylogenetic level-1 network from quartets.
\newblock {\em Bulletin of Mathematical Biology}, 76(10):2517--2541, 2014.

\bibitem{kelly2017lateral}
Luke~J Kelly and Geoff~K Nicholls.
\newblock {Lateral transfer in stochastic Dollo models}.
\newblock {\em The Annals of Applied Statistics}, 11(2):1146--1168, 2017.

\bibitem{McMahon-McMahon-2006}
A.~McMahon and R.~McMahon.
\newblock Why linguists don’t do dates: evidence from {Indo-European and
  Australian} languages.
\newblock In Peter Forster and Colin Renfrew, editors, {\em {Phylogenetic
  Methods and the Prehistory of Languages}}, pages 153--160. McDonald Institute
  for Archaeological Research Cambridge, 2006.

\bibitem{mirarab2014astral}
Siavash Mirarab, Rezwana Reaz, Md~S Bayzid, Th{\'e}o Zimmermann, M~Shel
  Swenson, and Tandy Warnow.
\newblock {ASTRAL:} genome-scale coalescent-based species tree estimation.
\newblock {\em Bioinformatics}, 30(17):i541--i548, 2014.

\bibitem{nakhleh2005perfect}
Luay Nakhleh, Don Ringe, and Tandy Warnow.
\newblock Perfect phylogenetic networks: A new methodology for reconstructing
  the evolutionary history of natural languages.
\newblock {\em Language}, pages 382--420, 2005.

\bibitem{nelson2011networks}
Shijulal Nelson-Sathi, Johann-Mattis List, Hans Geisler, Heiner Fangerau,
  Russell~D Gray, William Martin, and Tal Dagan.
\newblock Networks uncover hidden lexical borrowing in {Indo-European} language
  evolution.
\newblock {\em Proceedings of the Royal Society B: Biological Sciences},
  278(1713):1794--1803, 2011.

\bibitem{nichols2008tutorial}
Johanna Nichols and Tandy Warnow.
\newblock Tutorial on computational linguistic phylogeny.
\newblock {\em Language and Linguistics Compass}, 2(5):760--820, 2008.

\bibitem{reaz2014accurate}
Rezwana Reaz, Md~Shamsuzzoha Bayzid, and M~Sohel Rahman.
\newblock Accurate phylogenetic tree reconstruction from quartets: A heuristic
  approach.
\newblock {\em PloS one}, 9(8):e104008, 2014.

\bibitem{rexova2003cladistic}
Kate{\v{r}}ina Rexov{\'a}, Daniel Frynta, and Jan Zrzav{\`y}.
\newblock Cladistic analysis of languages: {Indo-European} classification based
  on lexicostatistical data.
\newblock {\em Cladistics}, 19(2):120--127, 2003.

\bibitem{ringe2002indo}
Don Ringe, Tandy Warnow, and Ann Taylor.
\newblock {Indo-European} and computational cladistics.
\newblock {\em Transactions of the Philological Society}, 100(1):59--129, 2002.

\bibitem{skelton2015borrowing}
Christina~Michelle Skelton.
\newblock Borrowing, character weighting, and preliminary cluster analysis in a
  phylogenetic analysis of the ancient {G}reek dialects.
\newblock {\em Indo-European Linguistics}, 3(1):84--117, 2015.

\bibitem{QMC}
Sagi Snir and Satish Rao.
\newblock {Quartet MaxCut:} a fast algorithm for amalgamating quartet trees.
\newblock {\em Molecular Phylogenetics and Evolution}, 62(1):1--8, 2012.

\bibitem{warnow2017computational}
Tandy Warnow.
\newblock {\em Computational phylogenetics: an introduction to designing
  methods for phylogeny estimation}.
\newblock Cambridge University Press, 2017.

\bibitem{warnow2006stochastic}
Tandy Warnow, Steven~N Evans, Donald Ringe, and Luay Nakhleh.
\newblock A stochastic model of language evolution that incorporates homoplasy
  and borrowing.
\newblock In Peter Forster and Colin Renfrew, editors, {\em {Phylogenetic
  Methods and the Prehistory of Languages}}, pages 75--90. McDonald Institute
  for Archaeological Research Cambridge, 2006.

\end{thebibliography}
\end{document}